\documentclass[reqno, 11pt]{amsart}
\usepackage{amssymb, natbib}
\usepackage[nesting]{hyperref}

\usepackage[pdftex]{graphicx}
\usepackage{listings}
\usepackage{multirow}
\usepackage{placeins}
\usepackage{color}
\usepackage{subfigure}
\usepackage{lscape}
\usepackage{dsfont}

\newcommand{\BlackBoxes}{\global\overfullrule5pt}

\BlackBoxes

\newcommand{\R}{\mathbb{R}}

\newcommand{\Eop}{\mathbb{E}}
\newcommand{\Pop}{\mathbb{P}}
\newcommand{\F}{\mathbb{F}}
\newcommand{\FC}{\mathcal{F}}

\newtheorem{theorem}{Theorem}

\newtheorem{lemma}[theorem]{Lemma}

\theoremstyle{definition}
\newtheorem{example}[theorem]{Example}

\numberwithin{equation}{section} \numberwithin{theorem}{section}

\def\0{\kern0pt\-\nobreak\hskip0pt\relax}

\makeatletter
\AtBeginDocument{%
 \def\@serieslogo{%
 \vbox to\headheight{%
 \parindent\z@ \fontsize{6}{7\p@}\selectfont
 \vss}}}

\def\makeoverbar#1#2#3#4#5#6#7{%
 \setbox0=\hbox{$\m@th#2\mkern#5mu{{}#3{}}\mkern#6mu$}%
 \setbox1=\null \dimen@=#4\fontdimen8#13 \dimen@=3.5\dimen@
 \advance\dimen@ by \ht0 \dimen@=-#7\dimen@ \advance\dimen@ by \wd0
 \ht1=\ht0 \dp1=\dp0 \wd1=\dimen@
 \dimen@=\fontdimen8#13 \fontdimen8#13=#4\fontdimen8#13
 \rlap{\hbox to \wd0{$\m@th\hss#2{\overline{\box1}}\mkern#5mu$}}
 \fontdimen8#13=\dimen@}

\def\mylabel#1#2{{\def\@currentlabel{#2}\label{#1}}}

\makeatother

\begin{document}


\makeatletter \providecommand\@dotsep{5} \makeatother

\title[Extremal Behavior of Long-Term Investors with Power Utility]{Extremal Behavior of Long-Term Investors with Power Utility}

\author[N. \smash{B\"auerle}]{Nicole B\"auerle${}^*$}
\address[N. B\"auerle]{Department of Mathematics,
Karlsruhe Institute of Technology, D-76128 Karlsruhe, Germany}

\email{\href{mailto:nicole.baeuerle@kit.edu}
{nicole.baeuerle@kit.edu}}

\author[S. \smash{Grether}]{Stefanie Grether${}^*$}
\address[S. Grether]{Department of Mathematics,
Karlsruhe Institute of Technology, D-76128 Karlsruhe, Germany}

\email{\href{mailto:stefanie.grether@kit.edu} {stefanie.grether@kit.edu}}

\thanks{${}^*$ Department of Mathematics,
Karlsruhe Institute of Technology, D-76128 Karlsruhe, Germany }

\begin{abstract}
We consider a Bayesian financial market with one bond and one stock where the aim is to maximize the expected power utility from terminal wealth. The solution of this problem is known, however there are some conjectures in the literature about the long-term behavior of the optimal strategy. In this paper we prove now that for positive coefficient in the power utility the long-term investor is very optimistic and behaves as if the best drift has been realized. In case the coefficient in the power utility is negative the long-term investor is very pessimistic and behaves as if the worst drift has been realized.

\end{abstract}
\maketitle

\vspace{0.5cm}
\begin{minipage}{14cm}
{\small
\begin{description}
\item[\rm \textsc{ Key words}]
{\small Bayes Approach, Investment Problem, Stochastic Ordering}
\end{description}
}
\end{minipage}

\section{Introduction}\label{sec:intro}\noindent
This paper investigates structural properties of the optimal portfolio choice in a financial market with one bond and one stock. The drift rate of
the stock price process is modeled as a random variable whose outcome is unknown to the investor. However, the investor is able to observe the stock
price process. Thus we face a Bayesian model. The aim is to maximize the expected power utility from terminal wealth.

Portfolio optimization problems with partial observation, in particular with unknown drift process have been studied extensively over the last decade.
The solution for the Bayesian case for power utility can be found among others in \cite{KZ01}, \cite{br05}, \cite{lm16}. This is a special case of a Hidden Markov model which has been treated in \cite{h03}, \cite{sh04}, \cite{br05} and with full information in \cite{br04}. More general and recent approaches towards these kind of models can be found in \cite{bjork2010}, \cite{shen2012}, \cite{fouque2015}, \cite{,gp16} where processes and utility functions differ. In this paper we simply take the solution for the optimal portfolio from the literature and discuss its behavior when the time horizon tends to infinity.  There have been some conjectures about this behavior in the literature. In \cite{cvitanic2006} in the case of a negative coefficient $\alpha$ for the power function it is reported that the long-term investor behaves more conservative. In \cite{lm16}, \cite{br05} it has been deduced from numerical data that in case $\alpha\in(0,1)$, the long-term investor behaves like an investor with known drift and the drift is the largest possible one. Meaning that the investor is quite optimistic. This observation stresses the fact that the behavior of an investor with power utility is rather different to one with logarithmic utility (see also recent findings in \cite{gp16}, Section 2.3.4), though it is known that the logarithmic utility can be obtained as limiting case from the power utility. In our paper we show that for $\alpha\in(0,1)$, indeed the optimal fraction invested in the stock converges for time horizon to infinity to the largest possible Merton ratio and thus investors are very optimistic. For $\alpha<0$, the point of view changes dramatically.  The optimal fraction invested in the stock converges for time horizon to infinity to the smallest possible Merton ratio. In this case investors are very pessimistic.  In particular the effect of learning does not play a role for a long-term investor in contrast to a short-term investor. For the influence of learning, see e.g. \cite{xia01}.

Our paper is organized as follows: In Section 2 we describe our model and recall the optimal portfolio strategy. In Section 3 we state and prove the behavior of a long-term investor which is quite surprising.

\section{The Investment Problem and the Optimal Solution}\label{sec:mod}\noindent
Suppose that $(\Omega, \FC, \F=\{\FC_t, 0\le t\le T\}, \Pop)$ is a filtered probability space and $T>0$ is a fixed time horizon. We
consider a financial market with one bond and one risky asset. The bond evolves according to
\begin{equation} d B_t = r B_t dt\end{equation}  with $r>0$ being the
interest rate. The stock price process $(S_t)$ is given by
\begin{equation}\label{stockprice} d S_t = S_t \left(\theta dt + \sigma d W_t\right) \end{equation}
where $(W_t)$ is an $\F$-Brownian motion, $\sigma>0$ and $\theta$ is a random variable with known initial distribution $\Pop(\theta=\mu_k) =: p_k>0,\;
k=1,\ldots ,d$ where $\mu_1,\ldots ,\mu_d$ are the possible values of $\theta$. We assume that $\theta$ and $(W_t)$
are independent. The outcome of $\theta$ is not known to the investor, but the investor is able to observe the stock price. Thus we face a Bayesian model. Note that this is a special case of a Hidden Markov Model where the value of the hidden process does not change.

The optimization problem is to find self-financing investment strategies in this market that maximize the expected utility from terminal wealth.
As utility function we choose the power utility $U(x)= \frac1\alpha x^\alpha$ with $\alpha<1,\alpha \neq 0$. The parameter $1-\alpha$ represents the risk aversion of the investor. Smaller $\alpha$ correspond to higher risk aversion.

Let $\F^S=\{\FC^S_t, 0\le t\le T\}$ be the filtration generated by the stock price process $(S_t)$.  In what follows we denote by $\pi_t\in\R$ the fraction of wealth invested in the stock at time $t$. $1-\pi_t$ is then the fraction of wealth invested in the bond at time $t$. If $\pi_t<0$, then this means that the
stock is sold short and $\pi_t>1$ corresponds to a credit. The process $\pi=(\pi_t)$ is called {\em portfolio strategy}. An admissible portfolio strategy has to be an $\F^S$-adapted process. The {\em wealth process} under an admissible portfolio strategy $\pi$ is given by the solution of the stochastic differential
equation
\begin{eqnarray}
d X_t^\pi &=& X_{t}^\pi \left[ (r+ (\theta-r)\pi_t)dt + \sigma \pi_t d W_t \right],
\end{eqnarray} where we assume that ${X}_0^\pi=x_0>0$ is the given initial
wealth. The {\em optimization problem } is  defined by
\begin{eqnarray}\label{eq:optprob}
&& \sup_{\pi} \Eop\left[\frac1\alpha\big({X}_T^\pi\big)^\alpha\right].
\end{eqnarray}
A portfolio strategy $\pi^\ast$ is {\em optimal} if it attains the supremum. The solution of this problem can e.g.\ be found in \cite{KZ01}, \cite{sh04}, \cite{br05}, \cite{lm16}.  We will only briefly sketch its derivation. Since $\theta$ is not known we have to estimate it by $\hat{\mu}_t := \Eop[\theta|\FC^S_t]$ which is the conditional expectation of $\theta$, given our observation up to time $t$. Define now $\widehat{W}_t := W_t +\frac1\sigma \int_0^t (\theta- \hat{\mu}_s) ds$. It can be shown that $(\widehat{W}_t)$ is an $\F^S$-Brownian motion (see Lemma 1 in Rieder \& B\"auerle 2005). Thus, the process
\begin{equation}Y_t := W_t + \frac{\theta-r}{\sigma}t= \widehat{W}_t +\frac1\sigma \int_0^t (\hat{\mu}_s-r)ds\end{equation}
is $\F^S$-adapted and hence observable for the investor. Finally we can represent our wealth process by
\begin{eqnarray}
d X_t^\pi &=& X_{t}^\pi \left[ (r+ (\hat{\mu}_t-r)\pi_t)dt + \sigma \pi_t d \widehat{W}_t \right].
\end{eqnarray} Note that $\hat{\mu}_t = \sum_{k=1}^d \mu_k \Pop(\theta=\mu_k | \FC^S_t)$ and the conditional probabilities $p_k(t):= \Pop(\theta=\mu_k | \FC^S_t)$ satisfy $p_k(0)=p_k$ and
\begin{eqnarray}
d p_k(t) &=& \frac1\sigma (\mu_k-\hat{\mu}_t) p_k(t) d \widehat{W}_t
\end{eqnarray}
which is a special case of the Wonham filter equation, see e.g. \cite{eam94}. The optimization problem is thus reduced to a situation with complete observation and can be solved with standard techniques like e.g.\ the Hamilton-Jacobi-Bellman equation.

In order to present the optimal portfolio strategy, let us introduce the following abbreviations:
Define $\gamma_k := \sigma^{-1}(\mu_k-r)$ and
\begin{equation}L_t(\mu_k,y) := \left\{ \begin{array}{cl} \exp(\gamma_k y- \frac{1}{2} \gamma_k^2t)& \mbox{ for } t\in (0,T],\\ 1 & \mbox{ for } t=0\end{array}\right.
\end{equation}
for $y\in\R,$ and $k=1,\ldots ,d.$ Further set
\begin{equation} F(t,y) := \sum_{k=1}^d L_t(\mu_k,y) p_k.\end{equation}
It is well-known that $(L_t^{-1}(\theta,Y_t))$ is a martingale density process with respect to the filtration $\F^{\theta,W}$ which is the filtration generated by $\theta$ and $(W_t)$. Then we can define a new probability measure $\mathbb{Q}$ by $d\mathbb{Q}/d\Pop= L_T^{-1}(\theta,Y_T)$. Under $\mathbb{Q}$ the process $(Y_t)$ is an $\F^{\theta,W}$-Brownian motion.  The process $(L_t(\theta,Y_t))$ is a $\mathbb{Q}$-martingale with respect to $\F^{\theta,W}$. Note that it can be shown that $\theta$ and $(Y_t)$ are independent under $\mathbb{Q}$. Then the Bayes formula implies that
\begin{equation}
    \Eop[1_{[\theta=\mu_k]}|\FC_t^S] = \frac{\Eop_\mathbb{Q}\Big[1_{[\theta=\mu_k]} L_T(\theta,Y_T)|\FC_t^S \Big]}{\Eop_\mathbb{Q}\Big[L_T(\theta,Y_T)|\FC_t^S \Big]}
\end{equation}
which yields that
\begin{eqnarray}
p_k(t) &=& \Pop(\theta=\mu_k | \FC^S_t) = \Pop(\theta=\mu_k | Y_t) = \frac{L_t(\mu_k,Y_t)p_k}{F(t,Y_t)}.
\end{eqnarray}

The proof of the following theorem can be found in \cite{KZ01} (Theorem 3.2 and Example 3.5) and \cite{br05} (Theorem 8).

\begin{theorem}\label{theo:sol}
The optimal portfolio strategy $(\pi_t^\ast)$ for problem \eqref{eq:optprob} is given in feedback form $\pi_t^\ast = u^\ast(t,T,Y_t)$, with
\begin{equation} u^\ast(t,T,y) = \frac{1}{\sigma(1-\alpha)} \frac{\Eop\left[ F(T,y+W_{T-t})^{\frac{\alpha}{1-\alpha}}\Big(\sum_{k=1}^d p_k \gamma_k L_T(\mu_k,y+W_{T-t}) \Big)\right]}{\Eop\left[F(T,y+W_{T-t})^{\frac{1}{1-\alpha}}\right]}\end{equation}
for $t\in[0,T], y\in\R$.
\end{theorem}

Note that a similar statement is true when the initial distribution of $\theta$ is more general, see e.g.\ \cite{KZ01}, \cite{lm16}. Further it is crucial to remark that in case $\theta$ is known and equal to $\mu$, then the optimal fraction to invest in the stock is equal to
\begin{equation} \frac{1}{(1-\alpha)}\frac{\mu-r}{\sigma^2}\end{equation}
independent of time and wealth. This is the so-called {\em Merton-ratio}. The representation of the optimal portfolio strategy in Theorem \ref{theo:sol} is derived by using a change of measure technique (for more details see the proof of Theorem 8 in Rieder \& B\"auerle 2005). An equivalent formulation would be
\begin{equation} u^\ast(t,T,y) = \frac{1}{1-\alpha} \frac{\Eop[\theta |Y_t=y]-r}{\sigma^2}+ h(t,T,y)\end{equation}
where the first term corresponds to the Merton-ratio where we replace the unknown $\theta$ by its estimate and $h(t,T,y)$ is the so-called hedging demand. The hedging demand vanishes for $t\to T$ and for $\alpha\to 0$ (the case $\alpha =0$ formally corresponds to the logarithmic utility). For a proof of the latter statements see Theorem 9 in Rieder \& B\"auerle (2005).

\section{Convergence of the Optimal Investment Strategy as Time Horizon tends to Infinity}\label{sec:MVFCFS}
\subsection{The case $\alpha \in(0,1)$}
In this section we are going to prove the following result:

\begin{theorem}\label{theo:main}
Let $t\ge 0$, $y\in\R$, $\alpha \in (0,1)$ and $r<\mu_1<\ldots <\mu_d$. Then
\begin{equation} \lim_{T\to\infty} u^\ast(t,T,y) = \frac{1}{\sigma(1-\alpha)} \frac{\mu_d-r}{\sigma}= \frac{1}{\sigma(1-\alpha)} \gamma_d. \end{equation}
\end{theorem}

This result is quiet interesting because it says that with a very large time horizon, the optimal investment strategy in the Bayesian case  with $\alpha\in(0,1)$ is approximately the same as in a setting with known, maximal possible drift. This result is independent of the prior distribution and the current belief, so the actual probability that the maximal drift has realized could be rather small. This means that the Bayes investor  with $\alpha\in(0,1)$ is pretty optimistic in the long run. Such a behavior has been conjectured in \cite{br05} and \cite{lm16} from the observation of numerical data. It is also remarkable that in the case of a logarithmic utility function, the optimal fraction of wealth invested in the stock is in feedback form given by (see e.g. Rieder \& B\"auerle 2005)
\begin{equation} u^*(t,T,y)= \frac{1}{\sigma^2(1-\alpha)} \Big({\frac{\sum_{k=1}^d p_k\mu_k L_t(\mu_k,y)}{F(t,y)}-r}\Big).\end{equation}
This expression is independent of the time horizon $T$. And it has been shown in \cite{br05} that for $\alpha\to 0$, the optimal fraction converges to the optimal fraction of the logarithmic utility function. This is of course no contradiction but shows that there is a significant difference of behavior between an investor with power utility and an investor with logarithmic utility function. In the next subsection we will discuss the case $\alpha<0$.

For the proof we need the following result, where $\varphi_t$ is the density of a normal distribution with zero mean and variance $t$.

\begin{lemma}\label{lem:monotone}
The function
\begin{equation} \alpha \mapsto \frac{\int_{\R} F(T,y+x)^{\frac{\alpha}{1-\alpha}} p_d L_T(\mu_d,y+x) \varphi_{T-t}(x)dx}{\int_{\R} F(T,y+x)^{\frac{1}{1-\alpha}}\varphi_{T-t}(x)dx}\end{equation}
is increasing for $\alpha<1$.
\end{lemma}

\begin{proof}
Similar proofs can be found in \cite{br05} Theorem 9 and \cite{lm16} Theorem 6.
We start by rewriting this expression as
\begin{equation}  \frac{\int_{\R} F(T,y+x)^{\frac{1}{1-\alpha}} \varphi_{T-t}(x) p_d {L_T(\mu_d,y+x)}{F(T,y+x)^{-1}} dx}{\int_{\R} F(T,y+x)^{\frac{1}{1-\alpha}}\varphi_{T-t}(x)dx}.\end{equation}
Now denote
\begin{equation} h_T(x) := p_d\frac{L_T(\mu_d,y+x)}{F(T,y+x)}.\end{equation}
The derivative of $h_T$ is given by
\begin{equation} h'_T(x) = \frac{\sum_{k=1}^d p_kp_d \exp\big((\gamma_k+\gamma_d)(y+x)-\frac12 T (\gamma_k^2+\gamma_d^2)\big)(\gamma_d-\gamma_k)}{F^2(T,y+x)}>0\end{equation}
for all $x\in\R$ and is positive by definition of $\gamma_d$. Thus, $h_T$ is increasing.

Next consider the density
\begin{equation} q_T(z,\alpha) := \frac{ F(T,y+z)^{\frac{1}{1-\alpha}} \varphi_{T-t}(z) }{\int_{\R} F(T,y+x)^{\frac{1}{1-\alpha}}\varphi_{T-t}(x)dx}\end{equation}
We show that this density is increasing in the likelihood ratio ordering with respect to $\alpha$. For more details about the likelihood ratio ordering consult \cite{ms}.
In order to do so, we have to show by definition that for $\alpha_1<\alpha_2<1$, the ratio
\begin{equation} z\mapsto \frac{q_T(z,\alpha_2)}{q_T(z,\alpha_1)}\end{equation} is increasing. We obtain
\begin{equation} \frac{q_T(z,\alpha_2)}{q_T(z,\alpha_1)}= F(T,y+z)^{\frac{1}{1-\alpha_2}-\frac{1}{1-\alpha_1}} \cdot C\end{equation}
where $C$ is a positive constant independent of $z$. Since $z\mapsto F(T,y+z)$ is increasing and we have by assumption ${\frac{1}{1-\alpha_2}-\frac{1}{1-\alpha_1}}>0$, the ratio is indeed increasing in $z$. Since the likelihood ratio ordering implies the usual stochastic ordering, the statement follows by definition of the stochastic ordering and the fact that $h_T$ is increasing.
\end{proof}

Now we are able to prove Theorem \ref{theo:main}.

\begin{proof}
It holds that
\begin{eqnarray}
\nonumber v^*(t,T,y)   &:=& \frac{\Eop\left[ F(T,y+W_{T-t})^{\frac{\alpha}{1-\alpha}}\Big(\sum_{k=1}^d p_k \gamma_k L_T(\mu_k,y+W_{T-t}) \Big)\right]}{\Eop\left[F(T,y+W_{T-t})^{\frac{1}{1-\alpha}}\right]} \\
   &=&  \frac{\int_{\R} F(T,y+x)^{\frac{\alpha}{1-\alpha}} \Big(\sum_{k=1}^d p_k \gamma_k L_T(\mu_k,y+x) \Big)\varphi_{T-t}(x)dx}{\int_{\R} F(T,y+x)^{\frac{1}{1-\alpha}}\varphi_{T-t}(x)dx}
\end{eqnarray}
and we have to show that $\lim_{T\to\infty} v^*(t,T,y) = \gamma_d$.

First define
\begin{equation} f_k(T,\alpha) := \frac{\int_{\R} F(T,y+x)^{\frac{\alpha}{1-\alpha}} \varphi_{T-t}(x) p_k L_T(\mu_k,y+x) dx}{\int_{\R} F(T,y+x)^{\frac{1}{1-\alpha}}\varphi_{T-t}(x)dx}.\end{equation}
Note that $f_k$ also depends on $y$. However $y$ will be fixed throughout and we do not make the dependence explicit in our notation.
Obviously $f_k(T,\alpha)>0$, $\sum_{k=1}^d f_k(T,\alpha)=1$ and
\begin{equation}v^*(t,T,y) = \sum_{k=1}^d \gamma_k f_k(T,\alpha).\end{equation} Thus, it is enough to show that $  \lim_{T\to\infty} f_d(T,\alpha) = 1$. Once we have shown this statement for an $\alpha'>0$, it holds for all $\alpha\ge \alpha'$, since by Lemma \ref{lem:monotone} we have that $f_d(T,\alpha') \le f_d(T,\alpha) \le 1$ for all $T>0$.
Now suppose $\alpha\in (0,1)$ and define $\beta := \frac{1}{1-\alpha}>1$. Consider the following lines of inequalities where we use the formula for the moment generating function of a normal distributed random variable in the last equation.
\begin{eqnarray}
\nonumber  f_d(T,\alpha) &\ge & \frac{\int_{\R} \Big(p_d L_T(\mu_d,y+x) \Big)^{\frac{\alpha}{1-\alpha}}\varphi_{T-t}(x) p_d L_T(\mu_d,y+x)dx}{\int_{\R} F(T,y+x)^\beta\varphi_{T-t}(x)dx}  \\
\nonumber   &=& \frac{\int_{\R} \Big(p_d L_T(\mu_d,y+x) \Big)^\beta\varphi_{T-t}(x) dx}{\int_{\R} F(T,y+x)^\beta\varphi_{T-t}(x)dx} \\
\nonumber   &=&  \frac{p_d^\beta \int_{\R} \exp\big(\gamma_d\beta (y+x)-\frac12 \gamma_d^2\beta T \big)\varphi_{T-t}(x) dx}{\int_{\R} F(T,y+x)^\beta\varphi_{T-t}(x)dx}  \\
   &=&  \frac{p_d^\beta \exp\big(\frac12 \gamma_d^2\beta(\beta-1) T \big)\cdot \exp\big(\gamma_d \beta y -\frac12 \gamma_d^2 t\beta^2\big)}{\int_{\R} F(T,y+x)^\beta\varphi_{T-t}(x)dx}.
\end{eqnarray}
Next we obtain with the Jensen inequality (note that $x\mapsto x^\beta$ is convex for $x\ge 0$):
\begin{eqnarray}
\nonumber  f_d(T,\alpha) &\ge & \frac{p_d^\beta \exp\big(\frac12 \gamma_d^2\beta(\beta-1) T \big)\cdot \exp\big(\gamma_d \beta y -\frac12 \gamma_d^2 t\beta^2\big)}{\int_{\R} \Big(\sum_{k=1}^d p_k L_T(\mu_k,y+x) \Big)^\beta\varphi_{T-t}(x)dx}\\
\nonumber  &\ge & \frac{p_d^\beta \exp\big(\frac12 \gamma_d^2\beta(\beta-1) T \big)\cdot \exp\big(\gamma_d \beta y -\frac12 \gamma_d^2 t\beta^2\big)}{\int_{\R} \Big(\sum_{k=1}^d p_k L_T(\mu_k,y+x)^\beta\Big)\varphi_{T-t}(x)dx}\\
  &=& \frac{p_d^\beta \exp\big(\frac12 \gamma_d^2\beta(\beta-1) T \big)\cdot \exp\big(\gamma_d \beta y -\frac12 \gamma_d^2 t\beta^2\big)}{\sum_{k=1}^d p_k \exp\big(\frac12 \gamma_k^2\beta(\beta-1) T\big)\cdot \exp\big(\gamma_k \beta y -\frac12 \gamma_k^2 t\beta^2\big)}.
\end{eqnarray}
Thus we obtain that
\begin{equation}\lim_{T\to\infty} f_d(T,\alpha) \ge  p_d^\beta \Big( \lim_{T\to\infty} \sum_{k=1}^d p_k \exp\big(\frac12 (\gamma_k^2-\gamma_d^2)\beta(\beta-1) T\big)\Big)^{-1} = p_d^{\beta-1}.\end{equation}
Since $\beta$ can be chosen arbitrarily close to $1$, the statement follows.
\end{proof}

\subsection{The case $\alpha <0$}
In this section we are going to prove the following result:

\begin{theorem}\label{theo:main}
Let $t\ge 0$, $y\in\R$, $\alpha <0$ and $r<\mu_1<\ldots <\mu_d$. Then
\begin{equation} \lim_{T\to\infty} u^\ast(t,T,y) = \frac{1}{\sigma(1-\alpha)} \frac{\mu_1-r}{\sigma}= \frac{1}{\sigma(1-\alpha)} \gamma_1. \end{equation}
\end{theorem}

In the case $\alpha<0$  the optimal investment strategy in the Bayesian case for an investor with a very large time horizon is approximately the same as in a setting with known, minimal possible drift. This result is independent of the prior distribution and the current belief. This means that the Bayes investor  with $\alpha<0$ is pretty pessimistic in the long run. Obviously there is a significant difference in the long run behavior of investors with $\alpha<0$, $\alpha=0$ (logarithmic utility) and investors with $\alpha >0$.

\begin{proof}
We define $v^*(t,T,y)$ and $f_k(T,\alpha)$ as in the previous proof and have to show that \linebreak $\lim_{T\to\infty} v^*(t,T,y)=\gamma_1$. For this it is enough to show that $\lim_{T\to\infty} f_1(T,\alpha)=1$. Similar to Lemma \ref{lem:monotone} it is possible to show that $\alpha \mapsto f_1(T,\alpha)$ is decreasing in $\alpha$ for all $\alpha<1$ (we simply have to replace $\gamma_d$ by $\gamma_1$). Thus we have for $\alpha'\le \alpha <0$ that $1\ge f_1(T,\alpha')\ge f_1(T,\alpha)$ and it is enough to prove $\lim_{T\to\infty} f_1(T,\alpha)=1$ for an $\alpha<0$ arbitrary close to zero. In order to do so we first rewrite $f_1$ as follows:
\begin{eqnarray}
\nonumber && f_1(T,\alpha)\\
 \nonumber &=& \frac{\int_{\R} F(T,y+x)^{\frac{\alpha}{1-\alpha}} \varphi_{T-t}(x) p_1 L_T(\mu_1,y+x) dx}{\int_{\R} F(T,y+x)^{\frac{1}{1-\alpha}}\varphi_{T-t}(x)dx} \\
\nonumber   &=& \frac{\int_{\R} F(T,y+x)^{\frac{1}{1-\alpha}} \varphi_{T-t}(x) p_1 L_T(\mu_1,y+x)F(T,x+y)^{-1} dx}{\int_{\R} F(T,y+x)^{\frac{1}{1-\alpha}}\varphi_{T-t}(x)dx} \\
   &=& \frac{\int_{\R} F(T,y+x\sqrt{T-t})^{\frac{1}{1-\alpha}} \varphi_1(x) p_1 L_T(\mu_1,y+x\sqrt{T-t})F(T,y+x\sqrt{T-t})^{-1} dx}{\int_{\R} F(T,y+x\sqrt{T-t})^{\frac{1}{1-\alpha}}\varphi_1(x)dx}.
\end{eqnarray}
We used the change of variables $z := \frac{x}{\sqrt{T-t}}$ in the last equation. Let us now consider the density
\begin{equation} f_T(z) := \frac{ F(T,y+z\sqrt{T-t})^{\frac{1}{1-\alpha}} \varphi_1(z) }{\int_{\R} F(T,y+x\sqrt{T-t})^{\frac{1}{1-\alpha}}\varphi_1(x)dx}\end{equation}
for fixed $y$ and $T$ and the function
\begin{equation}g(x,T) := \frac{p_1 L_T(\mu_1,y+x\sqrt{T-t})}{F(T,y+x\sqrt{T-t})}.\end{equation}
Then we can write $f_1(T,\alpha) = \int_\R g(x,T) f_T(x) dx$.
We further use the notation $\varphi_{k,T}$ for the density of a normal distribution $\mathcal{N}\big((\gamma_k \sqrt{T-t})(1-\alpha)^{-1}, (1-\alpha)^{-1} \Big).$ We obtain:
\begin{eqnarray}
\nonumber  && F(T,y+x\sqrt{T-t})^{\frac{1}{1-\alpha}} \varphi_1(x) =\\
 \nonumber &=&\Big(\sum_{k=1}^d p_k \exp\Big(\gamma_k x\sqrt{T-t}-\frac12 \gamma_k^2 T-\frac12 x^2 (1-\alpha)+\gamma_ky\Big) \Big)^\frac{1}{1-\alpha}\frac{1}{\sqrt{2\pi}} \\
\nonumber   &=& \Big(\sum_{k=1}^d p_k \varphi_{k,T}(x) \frac{\sqrt{2\pi}}{\sqrt{1-\alpha}}\exp\Big(\frac12 \gamma_k^2 (T\alpha-t)(1-\alpha)^{-1}+\gamma_ky\Big) \Big)^\frac{1}{1-\alpha}\frac{1}{\sqrt{2\pi}}  \\
   &=& \Big(\sum_{k=1}^d q_k(T) \varphi_{k,T}(x) \Big)^\frac{1}{1-\alpha}\frac{1}{\sqrt{2\pi}}
\end{eqnarray}
with $q_k(T) := p_k \exp(\frac12 \gamma_k^2 (T\alpha-t)(1-\alpha)^{-1}+\gamma_ky)  \frac{\sqrt{2\pi}}{\sqrt{1-\alpha}}$.
Now let \begin{equation}\hat{p}_k(T) := \frac{q_k(T)}{\sum_{j=1}^d q_j(T)}.\end{equation}
Then we can finally write the density $f_T$ as a mixture of normal densities:
\begin{equation} f_T(x) = \frac{\Big(\sum_{k=1}^d \hat{p}_k(T) \varphi_{k,T}(x) \Big)^\frac{1}{1-\alpha}}{\int_\R \Big(\sum_{k=1}^d \hat{p}_k(T) \varphi_{k,T}(x)\Big)^\frac{1}{1-\alpha}dx}.\end{equation}
Next we obtain the following estimate with the Jensen inequality since $x\mapsto x^\frac{1}{1-\alpha}$ is concave
\begin{eqnarray}
\nonumber  f_1(T,\alpha) &\ge& \frac{\int_\R \Big(\hat{p}_1(T) \varphi_{1,T}(x) \Big)^\frac{1}{1-\alpha}g(x,T)dx}{\int_\R \Big(\sum_{k=1}^d \hat{p}_k(T) \varphi_{k,T}(x)\Big)^\frac{1}{1-\alpha}dx} \ge  \frac{\int_\R \Big(\hat{p}_1(T) \varphi_{1,T}(x) \Big)^\frac{1}{1-\alpha}g(x,T)dx}{ \Big(\int_\R \sum_{k=1}^d \hat{p}_k(T) \varphi_{k,T}(x)dx \Big)^\frac{1}{1-\alpha}}  \\
   &=&  \int_\R \Big(\hat{p}_1(T) \varphi_{1,T}(x) \Big)^\frac{1}{1-\alpha}g(x,T)dx.
\end{eqnarray}
Next let us choose $\lambda\in\R$ such that $1<\lambda< \frac12 (\frac{\gamma_2}{\gamma_1}+1)$ which is possible due to our assumptions and consider the estimate
\begin{eqnarray}
\nonumber f_1(T,\alpha)&\ge&  \hat{p}_1(T)^\frac{1}{1-\alpha} \int_{-\infty}^{\gamma_1\lambda\sqrt{T-t}} \varphi_{1,T}(x)^\frac{1}{1-\alpha}g(x,T)dx\\
\nonumber &\ge& \hat{p}_1(T)^\frac{1}{1-\alpha} \int_{-\infty}^{\gamma_1\lambda\sqrt{T-t}} \varphi_{1,T}(x)g(x,T)dx\\
&\ge& \hat{p}_1(T)^\frac{1}{1-\alpha} g(\gamma_1\lambda\sqrt{T-t},T) \int_{-\infty}^{\gamma_1\lambda\sqrt{T-t}} \varphi_{1,T}(x)dx
\end{eqnarray}
where the second inequality follows form the fact that $\varphi_{1,T}(x)\le 1$ for all $x\in\R$ and the last inequality follows since $x\mapsto g(x,T)$ is decreasing which can easily be seen by inspecting its derivative (see also Lemma \ref{lem:monotone}). Finally we investigate what happens with these terms when $T\to\infty$. First we obtain
\begin{eqnarray}
\nonumber  \lim_{T\to\infty} \hat{p}_1(T) &=& \lim_{T\to\infty} \frac{p_1 \exp\Big( \frac12 \gamma_1^2  (T\alpha-t)(1-\alpha)^{-1}+\gamma_1 y\Big)}{\sum_{j=1}^d p_j \exp\Big( \frac12 \gamma_j^2  (T\alpha-t)(1-\alpha)^{-1}+\gamma_j y\Big)}  \\
\nonumber  &=& \lim_{T\to\infty} p_1\Big(\sum_{j=1}^d p_j \exp\Big( \frac12 (\gamma_j^2 -\gamma_1^2) (T\alpha-t)(1-\alpha)^{-1}+y(\gamma_j-\gamma_1)\Big)\Big)^{-1}\\
  & =&1,
\end{eqnarray}
since $\alpha<0$. Next we obtain
\begin{eqnarray}
\nonumber &&  \lim_{T\to\infty}g(\gamma_1\lambda\sqrt{T-t},T) \\
\nonumber  &=&\lim_{T\to\infty} p_1\Big(\sum_{j=1}^d p_j \exp\Big( T (\gamma_j -\gamma_1)(\lambda \gamma_1 -\frac12(\gamma_j+\gamma_1))-(t\lambda \gamma_1-y)(\gamma_j-\gamma_1)\Big)\Big)^{-1}\\
  & =&1,
\end{eqnarray}
since $\lambda$ is chosen such that $\lambda \gamma_1 -\frac12(\gamma_j+\gamma_1)<0$ for $j=2,\ldots,d$. Last but not least we obtain by change of variables that
\begin{eqnarray}
  \lim_{T\to\infty}\int_{-\infty}^{\gamma_1\lambda\sqrt{T-t}} \varphi_{1,T}(x)dx &=&  \lim_{T\to\infty}\int_{-\infty}^{\gamma_1\sqrt{T-t}(\lambda-\frac{1}{1-\alpha})\sqrt{1-\alpha}} \varphi_1(x)dx =1
\end{eqnarray}
since $\frac{1}{1-\alpha}<1<\lambda$. Finally note that $\alpha$ can be made arbitrarily small such that $\hat{p}_1(T)^\frac{1}{1-\alpha}$ can be made arbitrarily close to $1$ which implies the result.
\end{proof}

\begin{example}
Since we have explicit formulas it is easy to compute the optimal fraction of wealth which has to be invested in the stock. To illustrate our theoretical results we have done this calculation for a toy example with $\gamma_1 =1, \gamma_2 = 2, \gamma_3=3$, $p_1=p_2=0.3$, $p_3=0.4$, $\sigma=1$, $t=0$ and $y=0$. For $\alpha$ we have chosen $\alpha=0.5$ and $\alpha=-0.5$. In this example we have for $\alpha=0.5$ that $(\sigma(1-\alpha))^{-1} \gamma_d=6$ and for $\alpha=-0.5$ that $(\sigma(1-\alpha))^{-1} \gamma_1=2/3.$ We can see the convergence as $T\to\infty$ in the figures. The speed of convergence is rather different, but we have no statement about it.
\begin{figure}[h]
    \centering
     \includegraphics[height=5.5cm]{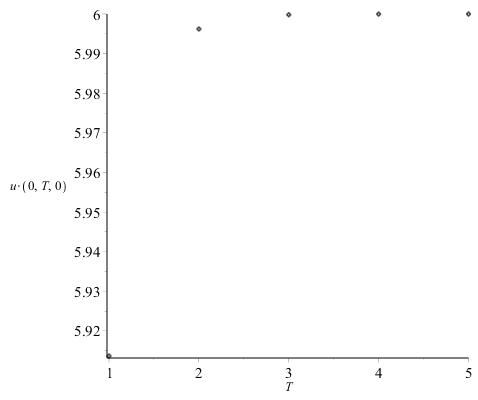}\hspace*{1.5cm} \includegraphics[height=5.5cm]{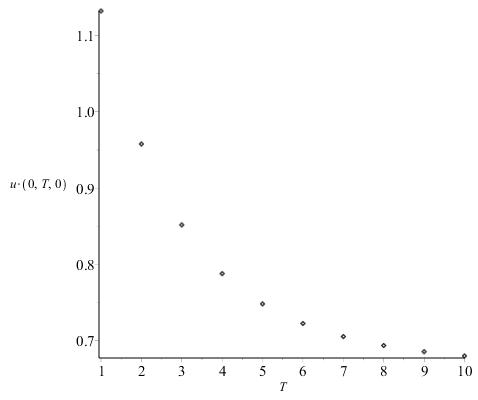}
   \caption{Optimal fraction $u^*(0,T,0)$ of wealth to invest in the stock as a function of the time horizon $T$ for $\alpha=0.5$ (left) and $\alpha=-0.5$ (right)}
    \label{fig:hp2}
\end{figure}
\end{example}

\section{Conclusion}
In this paper we have shown that investors with very long time horizon who cannot observe the drift of the stock and use a Bayesian model behave in an extreme way which is rather surprising. In practice the drift may not be constant over a long time horizon and it would be desirable to do the same analysis in e.g. a Hidden Markov Model. But there are no explicit formulas in this case which makes the analysis much harder. From a theoretical point of view the convergence analysis shows that the behavior of investors with power utility $U(x)=\frac1\alpha x^\alpha$ for $\alpha<1,\alpha\neq 0$ is very sensitive with respect to $\alpha$.



\bibliographystyle{abbrv}

\end{document}